\documentclass[%
pra,
twocolumn,
 amsmath,amssymb,
 aps,
]{revtex4-2}
\usepackage{graphicx}
\usepackage{epstopdf}
\usepackage{dcolumn}
\usepackage{amsmath}
\usepackage{amsthm}
\usepackage{hyperref}

\newtheorem{theorem}{Theorem}

\begin{document}

\preprint{APS/123-QED}

\title{Quantum multigraph states and multihypergraph states}
\author{Xiao-Dong Zhang$^{1}$}
\author{Bin-Bin Cai$^{1,2}$}
\altaffiliation{Corresponding author. Email address: cbb@finu.edu.cn}
\author{Song Lin$^{1}$}
\altaffiliation{Corresponding author. Email address: lins95@fjnu.edu.cn}
\affiliation{
	$^{1}$College of Computer and Cyber Security, Fujian Normal University, Fuzhou 350117, China\\
	$^{2}$Digital Fujian Internet-of-Things Laboratory of Environmental Monitoring, Fujian Normal University, Fuzhou 350117, China}
\date{\today}

\begin{abstract}
We proposed two classes of multiparticle entangled states, the multigraph states and multihypergraph states, defined by unique operations on the edges and hyperedges. A key discovery is the one-to-one correspondence between the proposed multihypergraph states and the generalized real equally weighted states when $d$ is prime. While for composite $d$, multihypergraph states form a subset of the generalized real equally weighted states. Meanwhile, we detailed a method for constructing real equally weighted states from hypergraph states and revealed the generalized real equally weighted states which cannot be generated from $d$-dimensional hypergraph states.

\end{abstract}
 

\maketitle


\section{Introduction}

Multiparticle entangled states play a crucial role in quantum computing and quantum information fields. Hillery et al. \cite{hillery1999quantum} pioneered a secret sharing scheme in 1999 using the Greenberger-Horne-Zeilinger state. In 2001, Briegel and Raussendorf \cite{briegel2001persistent} introduced the ``cluster state", a specific type of multiparticle entangled state, and subsequently proposed a new quantum computing approach based on cluster state measurements \cite{raussendorf2001one}. In the same year, the cluster states were used to design quantum error-correcting\cite{schlingemann2001quantum}. And, Schlingemann and Werner \cite{schlingemann2001quantum} firstly to employ graph representations to describe these quantum states. In 2003, Raussendorf et al. \cite{raussendorf2003measurement} formally coined the term ``graph state" in their exploration of measurement-based quantum computing. Graph states, characterized by connections through  controlled-$Z$ gates, facilitate the simulation of any unitary gate through specific measurement bases and sequences \cite{raussendorf2003measurement}. In 2013, Rossi et al. \cite{rossi2013quantum} and Qu et al. \cite{qu2013encoding} extended the graph states to the hypergraph states. And a one-to-one correspondence is discovered between hypergraph states and real equally weighted states (REWSs)\cite{rossi2013quantum,qu2013encoding}. Hypergraph states, notable for their ``Pauli universality" in measurement-based computation \cite{miller2016hierarchy}, enable the simulation of universal unitary gates solely through Pauli measurements of particles within the state. The class of multiparticle entangled states, graph states and hypergraph states, have been extensively studied\cite{hein2004multiparty,guhne2005bell,lu2007experimental,qu2013relationship,guhne2014entanglement,lyons2015local,gachechiladze2016extreme,ghio2017multipartite,morimae2017verification,zhu2019efficient,shettell2020graph,baccari2020scalable,zhou2022entanglement} and widely used in quantum computing\cite{childs2005unified,walther2005experimental,nielsen2006cluster,raussendorf2006fault,briegel2009measurement,broadbent2009universal,mantri2013optimal,hayashi2015verifiable,morimae2016measurement,fitzsimons2018post,gachechiladze2019changing} and quantum information fields\cite{markham2008graph,bell2014experimental,banerjee2020quantum,li2022efficient}. 

 The focus on graph states and hypergraph states within 2-dimensional quantum systems has garnered considerable scholarly interest, leading to an expansion of research into $d$-dimensional quantum systems. Studies have been conducted on $d$-dimensional graph states \cite{looi2008quantum,keet2010quantum,tang2013greenberger}, followed by investigations into $d$-dimensional hypergraph states \cite{steinhoff2017qudit,xiong2018qudit,malpetti2022multipartite}. Xiong et al. \cite{xiong2018qudit} elucidated the quantitative relationship between $d$-dimensional hypergraph states and generalized real equally weighted states (GREWSs), noting that $d$-dimensional hypergraph states are encompassed within GREWSs and are significantly fewer in number. The existing research on graph states and hypergraph states \cite{hein2004multiparty,guhne2005bell,lu2007experimental,shettell2020graph,baccari2020scalable,childs2005unified,walther2005experimental,nielsen2006cluster,raussendorf2006fault,markham2008graph,briegel2009measurement,broadbent2009universal,mantri2013optimal,bell2014experimental,hayashi2015verifiable,morimae2016measurement,fitzsimons2018post,rossi2013quantum,qu2013encoding,miller2016hierarchy,qu2013relationship,guhne2014entanglement,lyons2015local,gachechiladze2016extreme,ghio2017multipartite,morimae2017verification,zhu2019efficient,gachechiladze2019changing,banerjee2020quantum,zhou2022entanglement,looi2008quantum,keet2010quantum,tang2013greenberger,steinhoff2017qudit,xiong2018qudit,malpetti2022multipartite,li2022efficient} primarily employed simple graphs \cite{berge1973graphs,balakrishnan1997schaum} and weighted graphs \cite{chartrand2013first}. In this paper, we proposed, for the first time, the quantum state corresponding to the multigraph and multihypergraph, where multifarious edges can connect the same vertices and highlighted some of its advantageous properties.

The structure of this paper is outlined as follows. In Sec. \ref{sec:2}, the review of essential graph theory concepts and the definitions of $d$-dimensional graph states and hypergraph states are provided. In Sec. \ref{sec:3}, the definition of multigraph states and multihypergraph states is proposed. In Sec. \ref{sec:4}, the association between multihypergraph states and GREWSs is detailed. In Sec. \ref{sec:5}, a summary and an outlook on future research are concluded.
\section{\label{sec:2}PRELIMINARIES}

In this section, we succinctly overview the foundational aspects of graph theory pertinent to graph states and hypergraph states, along with their respective definitions.

\subsection{Fundamentals of graph theory}

Graph theory encompasses a diverse array of graph types, including simple graphs, weighted graphs, multigraphs, hypergraphs, weighted hypergraphs, and multihypergraphs \cite{berge1973graphs,balakrishnan1997schaum,chartrand2013first,bollobas1998modern}. A simple graph is characterized by edges that connect two vertices, with a maximum of one edge between any two vertices. Weighted graphs are an extension of simple graphs, assigning weights to each edge. Multigraphs come in two forms. One form is described that multiple edges connecting two vertices are identical. The same edge allowed to appear multiple times. And the multigraphs can be viewed as weighted graphs \cite{chartrand2013first}. The other form is considered that each edge as a distinct entity, akin to nodes. The multiple edges between two vertices are different, as exemplified by the electrical network multigraph in Ref. \cite{bollobas1998modern}. In hypergraphs, edges may connect any number of vertices, with those connected to a single vertex termed as rings \cite{berge1973graphs}. In hypergraphs field, the weighted hypergraphs and multihypergraphs also exist. And their definition are similar to those of weighted graphs and multigraphs, except that weighted hypergraphs and multihypergraphs can connect any number of vertices.

A graph consisting of $N$ vertices can be described as the set pair $G=(V, E)$, where $V=\mathbb{Z}_N$ denotes the set of vertices and $E$ denotes the set of edges. The $e(j,k)\in E$ represents the edge connecting vertices $j$ and $k$. Similarly, a hypergraph of $N$ vertices is represented as the set pair $\scriptstyle\widetilde{G}=\left\{V,\widetilde E\right\}$, with  $\scriptstyle\widetilde E\subseteq\wp(V)$ as the set of hyperedges. Here, $\wp(V)$ denotes the power set of $V$, and $e\left(v_0,v_1,\cdots,v_{t-1}\right)\in \widetilde E$ is the hyperedge connecting $t \in \mathbb{Z}_{N+1}$vertices, where $v_0 < v_1 < \cdots <v_{t-1} \in \mathbb{Z}_{N}$ indicate the connected $t$ vertices. Notably, when $t$ equals 0, $e\left(v_0,v_1,\cdots,v_{t-1}\right)$ refers to an empty edge. Simple graphs and hypergraphs can be transformed into weighted graphs through the addition of weights to their respective edges. Within such weighted graphs, the weight assigned to an edge $e$ is denoted as $m_e$. Consequently, the relationships among simple graphs, multigraphs, hypergraphs, and multihypergraphs are elucidated and can be visualized as depicted in Fig.\ref{fig1}.
\begin{figure}[h]%
	\hspace{-10mm}
	\includegraphics[width=0.45\textwidth]{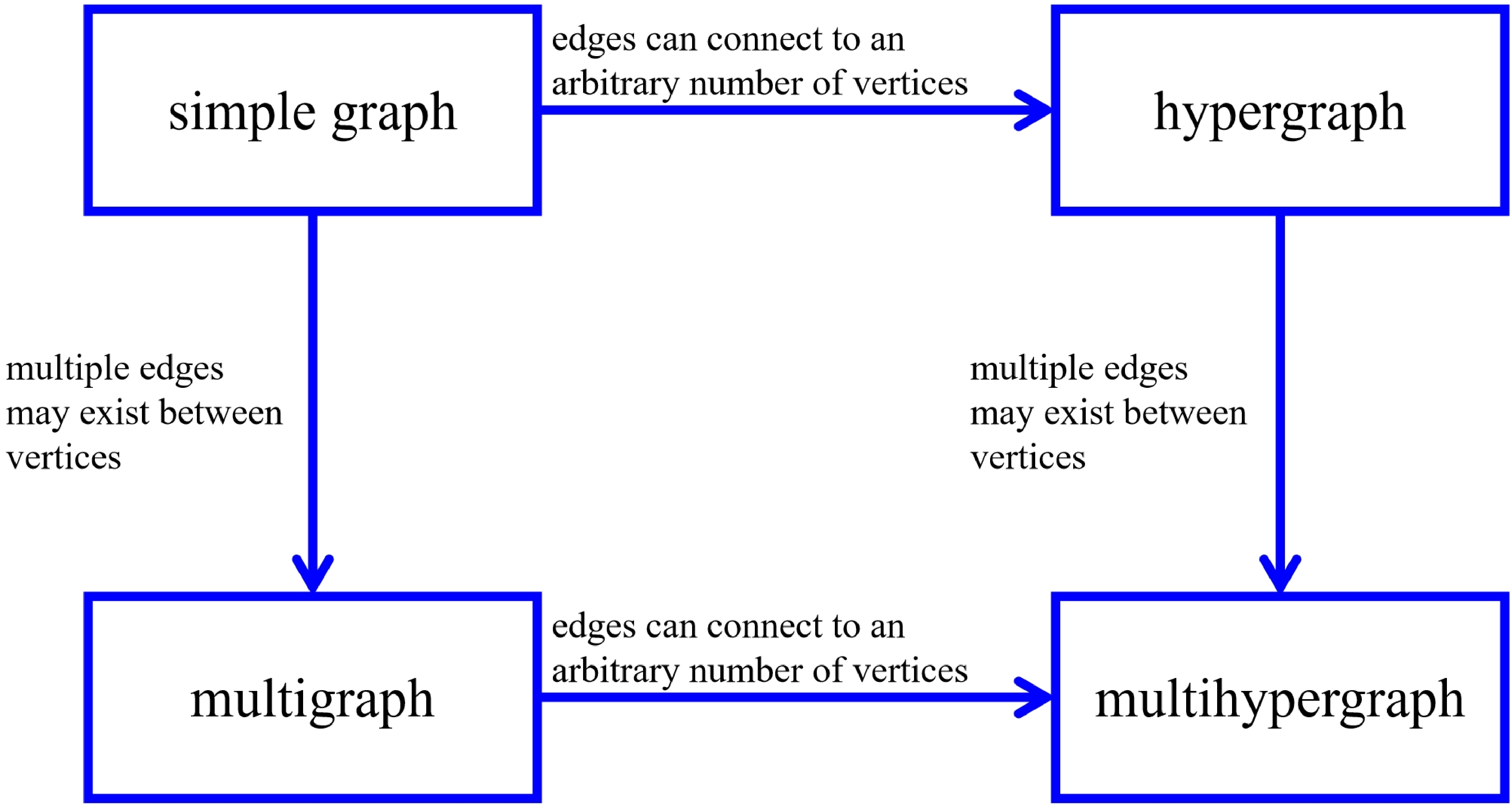}
	\caption{Relationships among simple graphs, multigraphs, hypergraphs, and multigraphs.}\label{fig1}
\end{figure}

\subsection{A review of graph states and hypergraph states}
\subsubsection{Graph States}
The concept of the graph state was formally introduced, as elucidated through research presented in Refs. \cite{briegel2001persistent,raussendorf2001one,schlingemann2001quantum,raussendorf2003measurement} . Initially termed as ``cluster states" \cite{briegel2001persistent}, these states were later renamed as ``graph states" \cite{schlingemann2001quantum,raussendorf2003measurement} following the adoption of graph theoretical representations to describe these quantum states. The construction of cluster states is depicted in Fig.\ref{fig2}, where dots symbolize the quantum state $\left| {+_d} \right\rangle  =  d^{-1/2}\sum_{i = 0}^{d-1}  \left| i \right\rangle $, solid lines represent  controlled-$Z$ gates, and dashed lines indicate additional $\left| +_d \right\rangle$ states and  controlled-$Z$ gates. In the graphical depiction of cluster states, the points and solid lines in Fig.\ref{fig2} correspond to the vertices and edges in graph theory. In subsequent discussions within this paper, we will exclusively use the term ``graph states" for clarity and consistency.
\begin{figure}[h]%
	\centering
	\includegraphics[width=0.35\textwidth]{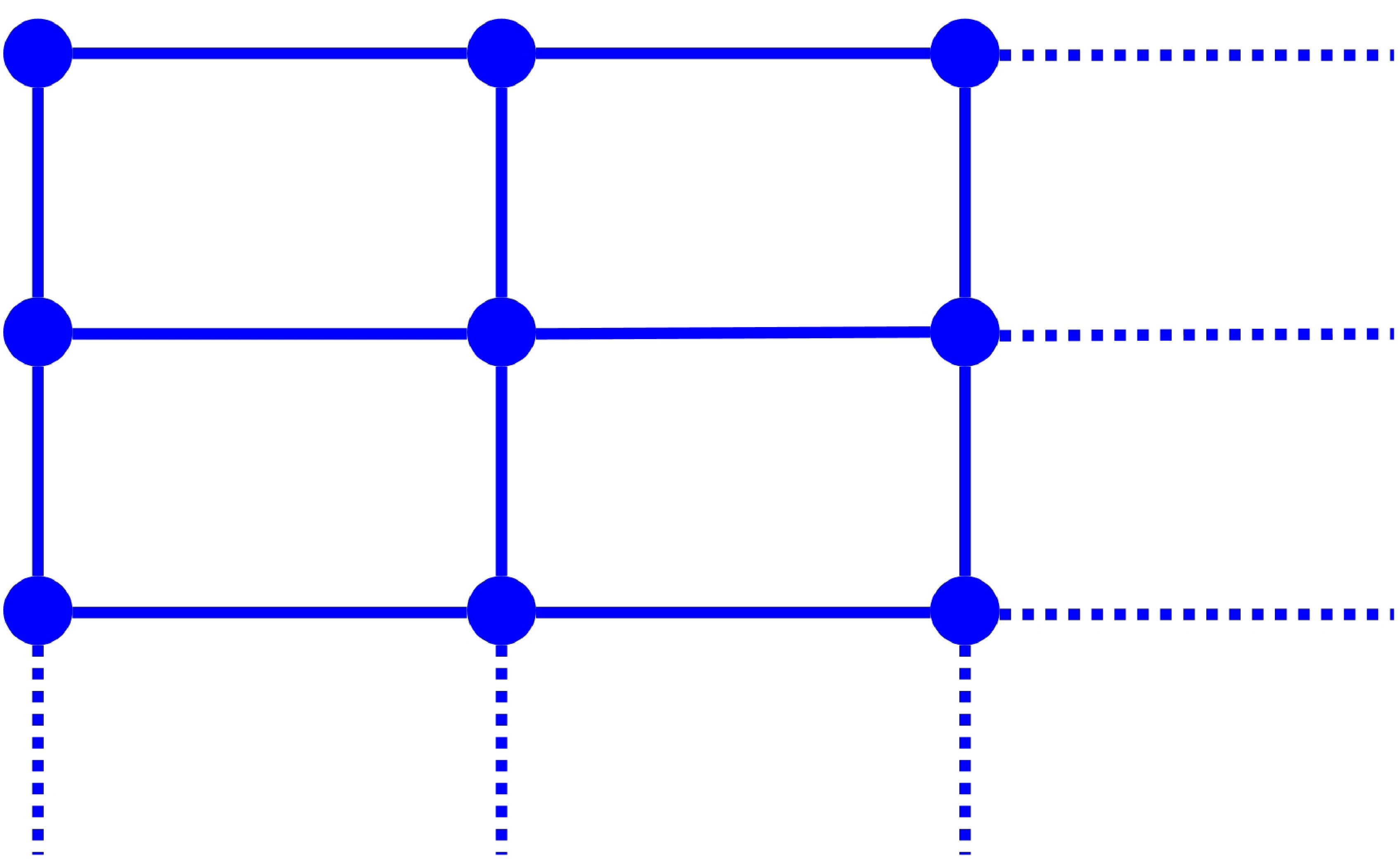}
	\caption{Cluster state. The dots denote quantum $\left| +_d \right\rangle$ states, solid lines correspond to  controlled-$Z$ gates, and dashed lines indicate additional $\left| +_d \right\rangle$ states and  controlled-$Z$ gates.}\label{fig2}
\end{figure}

In accordance with varying graph structures, distinct graph states can be formulated. Prior to delving into the description of such quantum states, it is imperative to define the vertices and edges within the graph from a quantum perspective. The formal representation of vertices and edges in graph states has been previously delineated in Refs. \cite{raussendorf2003measurement,looi2008quantum,hein2004multiparty}. For a graph $G=(V,E)$ comprising $N$ vertices, the corresponding graph state is characterized by $N$ vertices denoted as $\left|+_d\right\rangle^{\otimes N}=\scriptstyle\underbrace{\scriptstyle\left|+_d\right\rangle \otimes \left|+_d\right\rangle \otimes \cdots \otimes \left|+_d\right\rangle}_N$. The edge connecting vertices $j$ and $k$ is represented by
\begin{multline}
CZ_{e(j,k)}= \sum_{i_0, i_1, \cdots, i_{N-1} = 0}^{d-1}\omega_d^{i_j\times i_k}\\[-1mm]\cdot \left| i_0, i_1, \cdots, i_{N-1}\right\rangle \left\langle i_0, i_1,\cdots,i_{N-1} \right|,
\label{eq:1}
\end{multline}
where $\omega_d=e^{2\pi \mathbf{i}/d}$($\mathbf{i}=\sqrt{-1}$). The $d$-dimensional graph state corresponding to graph $G$ is defined as
\begin{equation}
	\left| G\right\rangle= \prod_{e\in E}{CZ_e}^{m_e} \left|+_d\right\rangle^{\otimes N},
	\label{eq:2}
\end{equation}
where $m_e\in\mathbb{Z}_d$ represents both the weight of an edge and the number of operations performed.
\subsubsection{Hypergraph States}
In 2013, the concept of the hypergraph state was introduced in Refs. \cite{rossi2013quantum,qu2013encoding}, which extended the framework of quantum graph states. Within a $d$-dimensional hypergraph state, vertices are defined as the $d$-dimensional single-particle states $\left|+_d\right\rangle=d^{-1/2}\sum_{i=0}^{d-1}\left| i\right\rangle$. The hyperedge connects vertices $v_0,\cdots,v_{t-1}$, is characterized by 
\begin{multline}
	{\widetilde{CZ}}_{e(v_0,v_1,\cdots,v_{t-1})}=\sum_{i_0,i_1,\cdots,i_{N-1}=0}^{d-1} \Bigl(  \omega_d^{i_{v_0}\times i_{v_1}\times \cdots \times i_{v_{t-1}}}\Bigr. \\[-1mm]
	\Bigl. \cdot \left| i_0, i_1,\cdots,i_{N-1}\right\rangle \left\langle i_0, i_1,\cdots,i_{N-1}\right| \Bigr).
	\label{eq:3}
\end{multline}
The $d$-dimensional hypergraph state corresponding to graph $\widetilde{G}=(V,\widetilde E)$ is defined as
\begin{equation}
	\left|\widetilde{G}\right\rangle=\left(\prod_{e\in \widetilde E}{{\widetilde{CZ}}_e}^{m_e}\right)\left|+_d\right\rangle^{\otimes N},
	\label{eq:4}
\end{equation}
where $\left|+_d\right\rangle^{\otimes N}$ refers to the $N$ vertices comprising the hypergraph state. As per the established definition of hyperedges, a $d$-dimensional hypergraph state comprises $2^N-1$ varieties of non-empty hyperedges, with each edge subject to $m_e\in\mathbb{Z}_d$ iterations of operations. Consequently, the count of non-trivial $d$-dimensional hypergraph states amounts to $d^{2^N-1}$, aligning with the findings reported in Refs. \cite{rossi2013quantum,qu2013encoding,xiong2018qudit}. When articulated within the framework of stabilizer systems, the stabilizer for a $d$-dimensional hypergraph state is characterized as follows
\begin{multline}
	\fontsize{8}{10}
	g_k=\left( \prod_{e\in\widetilde E} {{\widetilde{CZ}}_e}^{m_e}\right) X_k\left(\prod_{e\prime \in \widetilde E}{{\widetilde{CZ}}_{e\prime}}^{{d-m}_{e\prime}}\right)\\ = X_k\prod_{e \in\widetilde E,k\in e}{{\widetilde{CZ}}_{e \backslash \left\{k\right\}}}^{m_e(d-1)}.
	\label{eq:5}
\end{multline}
The proof of Eq. \eqref{eq:5} is delineated in the appendix \ref{ap:A}. Noting that be noted that the stabilizers presented here are equivalent to those described in Ref. \cite{xiong2018qudit}, even though they are distinct in form.
\section{\label{sec:3}The proposed multigraph states and multihypergraph states}
In this section, we proposed concepts of multigraph states and multihypergraph states. Traditional graph states and hypergraph states are typically constructed by a series of $Z$ gates and multiparticle  controlled-$Z$ gates. Our approach, however, involves a novel series of quantum gates specific to the construction of multigraph states and multihypergraph states. Initially, we describe a series of $d$-dimensional quantum gates, followed by an explanation of how these gates facilitate the construction of multigraph states and multihypergraph states. In the context of this paper, the multigraph and multihypergraph utilized differ in that various edges (or hyperedges) connecting the same vertices represent distinct entities. We employ a set of sequences related to vertices to encode these edges (hyperedges). In both the multigraph and the multihypergraph, $N$ integer variables $s_0,s_1,\cdots,s_{N-1}\in\mathbb{Z}_d^\ast$ ($\mathbb{Z}_d^\ast=\mathbb{Z}_d\backslash\left\{0\right\}$) are defined corresponding to the $N$ vertices. For an edge or hyperedge linking $t$ identical vertices, $v_0,v_1,\cdots,v_{t-1}$, is encoded using an integer sequence $s=(s_{v_0},s_{v_1},\cdots,s_{v_{t-1}})$. Such edges or hyperedges can be represented as $e\left(v_0,v_1,\cdots,v_{t-1},s\right)$. Consequently, a multigraph or multihypergraph is denoted as a set pair comprising vertices and edges, $\scriptstyle\widehat{G}=(V,\widehat E)$  or $\scriptstyle\widehat{\widetilde{G}}=(V,\widehat{\widetilde{E}})$. Here, $t=2$ characterizes a multigraph, whereas $t\in\mathbb{Z}_{N+1}$ typifies a multihypergraph.

In 2012, Howard and Vala \cite{howard2012qudit} explored the qudit analog of the qubit $\pi/8$ gate within the clifford hierarchy \cite{gottesman1999demonstrating}. They proposed a formulation for the qudit gate, presented below
\begin{equation}
	U_v=U_{\left(l_0,l_1,\cdots,l_{d-1}\right)}=\sum_{k=0}^{d-1}{\omega_d^{h_k}\left| k \right\rangle \left\langle k \right|}.
	\label{eq:6}
\end{equation}
The selection of $U_v$ based on $h=\left(h_0,h_1,\cdots,h_{d-1}\right)$ forms a series of quantum gates with notable applications in magic-state distillation \cite{campbell2012magic}. Define $h_k=\sum_{j=0}^{\eta}{a_jk^j}$, where $\eta\in\mathbb{Z}_d^*$, and $\sum_{j=1}^{\eta}{a_jk}^j$ represents a polynomial of degree $\eta$ in the integer residual ring $\mathbb{Z}_d$. For $h_k=k^2/2-k/2$ and $h_k=k^3/6$, $U_v$ corresponds to the $d$-dimensional $S$ gate and $\pi/8$ gate Refs. \cite{gottesman1998fault,prakash2018normal}, respectively. In this paper, $U_v$ is employed with $h_k=k^\eta$. Specifically, for $\eta=1$, $U_v=Z;$ for $\eta=2$, $U_v=ZSS$; and for $\eta=3$, $U_v=T^6$. The ensuing definitions of multigraph states and multihypergraph states are predicated on this quantum gate typology. Here, the quantum gate linked with the edges of the multigraph states and multihypergraph states serves as the  controlled gate of $U_v$ for 2-particle and multiparticle systems.
\subsection{Multigraph state}
In the multigraph state, the vertex is represented by a $d$-dimensional single-particle state $\left|+_d\right\rangle=d^{-1/2}\sum_{i=0}^{d-1}\left| i\right\rangle$. The edges within this framework are delineated through
\begin{multline}
	{\widehat{CZ}}_{e\left(j,k,s\right)}=\sum_{i_0,i_1,{\cdots,i}_{N-1}=0}^{d-1}\Bigl(\omega_d^{{i_j}^{s_j}\times{i_k}^{s_k}}\Bigr. \\[-2mm] \Bigl. \cdot \left|i_0,i_1,\cdots,i_{N-1}\right\rangle \left\langle i_0,i_1,{\cdots,i}_{N-1}\right|\Bigr).
	\label{eq:7}
\end{multline}
The multigraph state, corresponding to the multigraph $\widehat{G}=(V,\widehat E)$, can be represented by
\begin{equation}
	\left| \widehat{G}\right\rangle=\prod_{e\in\widehat E} {{\widehat{CZ}}_e}^{m_e} \left|+_d\right\rangle^{\otimes N},
	\label{eq:8}
\end{equation}
where $m_e\in\mathbb{Z}_d$ represents both the weight of an edge and the number of operations performed. We provide an illustrative example of a multigraph state with $N=3$ and $d=256$. Here, vertices $v_0,v_1,v_2$ symbolize the three primary colors, red, green and blue. Correspondingly, three integer variables, $s_0,s_1,s_2$, represent the RGB values of these primary colors as encoded in computer systems. For instance, the multihypergraph state $\scriptstyle\left|\widehat{G}\right\rangle={\widehat{CZ}}_{e_4}{\widehat{CZ}}_{e_3}{\widehat{CZ}}_{e_2}{\widehat{CZ}}_{e_1}{\left|+_{256}\right\rangle}^{\otimes3}$ correlates to the multihypergraph depicted in Fig.\ref{fig3}(a).
\begin{figure}[h]%
	\centering
	\includegraphics[width=0.45\textwidth]{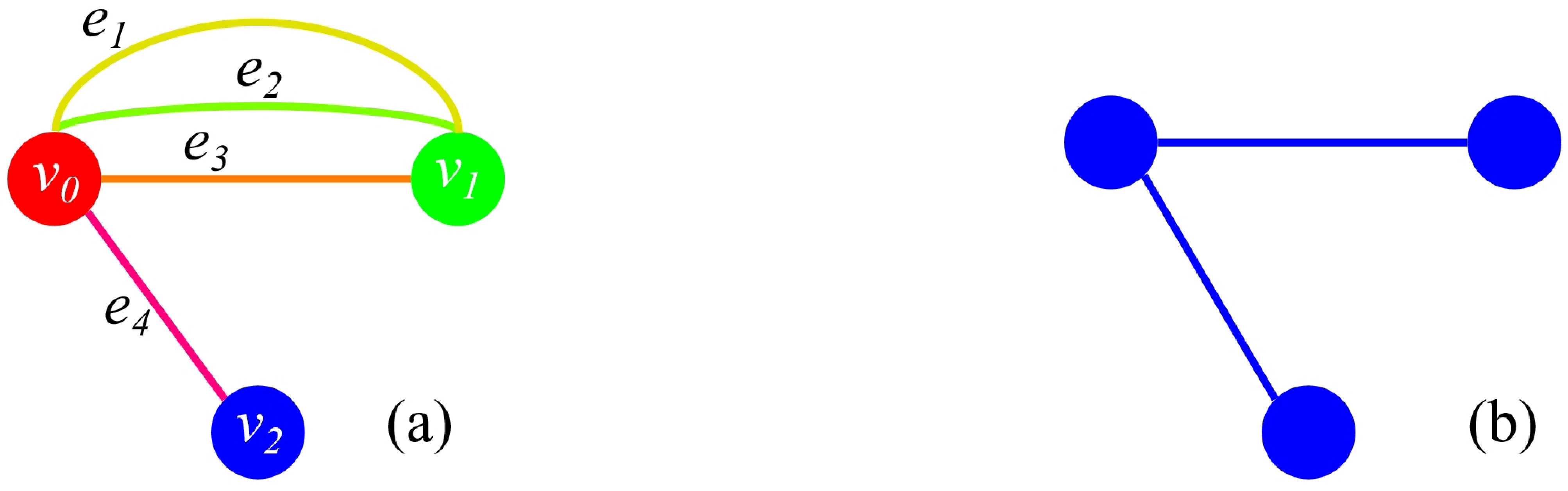}
	\caption{3-vertex multigraph and 3-vertex graph. (a) displays a multigraph with four edges, $e_1$ to $e_4$, each marked by unique RGB values (225,225,0), (125,255,0), (255,125,0) and (255,0,125). Conversely, (b) shows a graph with uniform coloring, where all vertices and edges are blue (0,0,255).}\label{fig3}
\end{figure}
\subsection{Multihypergraph state}
In the multihypergraph state, the hyperedges within this framework are delineated through
\begin{multline}
	\fontsize{8}{10}
	{\widehat{\widetilde{CZ}}}_{e(v_0,v_1,\cdots,v_{t-1},s)}=\sum_{i_0,i_1,{\cdots,i}_{N-1}=0}^{d-1}\Bigl( \omega_d^{{i_{v_0}}^{s_{v_0}}\times\cdots\times{i_{v_{t-1}}}^{s_{v_{t-1}}}}\Bigr. \\[-2mm] \Bigl. \cdot \left|i_0,i_1,{\cdots,i}_{N-1}\right\rangle \left\langle i_0,i_1,{\cdots,i}_{N-1}\right| \Bigr).
	\label{eq:9}
\end{multline}
The multihypergraph state, corresponding to the structure of multihypergraph $\scriptstyle\widehat{\widetilde{G}}=(V,\scriptstyle\widehat{\widetilde E})$, is characterized as
\begin{equation}
	\left|\widehat{\widetilde{G}}\right\rangle=\left(\prod_{e\in \widehat{\widetilde E}}{{\widehat{\widetilde{CZ}}}_e}^{m_e}\right)\left|+_d\right\rangle^{\otimes N}.
	\label{eq:10}
\end{equation}
It is straightforward to ascertain that the quantity of multiple hyperedges equals $d^{N-1}$. Each quantum operation associated with these hyperedges undergoes $d$ cycles within the $d$-dimensional quantum state, culminating in a total of $d^{d^{N-1}}$ multihypergraph states. The stabilizers of a multihypergraph state can be articulated as
\begin{multline}
	\fontsize{8}{10}
	g_k=\left(\prod_{e\in \widehat{\widetilde E}} {{\widehat{\widetilde{CZ}}}_e}^{m_e}\right) X_k\left(\prod_{e\prime \in \widehat{\widetilde E}}{{\widehat{\widetilde{CZ}}}_{e\prime}}^{{d-m}_{e\prime}}\right)\\ =X_k\prod_{e\in \widehat{\widetilde E},k\in e}{{\widehat{\widetilde{CZ}}}_{e\backslash \left\{k\right\}}}^{m_e(d-1)}.
	\label{eq:11}
\end{multline}
The proof of Eq. \eqref{eq:11} is delineated in the appendix \ref{ap:B}. In continuation of the aforementioned concepts, we consider an example of a multihypergraph state with $N=3$ and $d=256$. The vertices $v_0,v_1,v_2$ are assigned to represent the three primary colors, red, green and blue. The integer variables $s_0,s_1,s_2$ correspond to the RGB values of these colors as they are represented in computer systems. Specifically, the multihypergraph state $\scriptstyle\left|\widehat{\widetilde{G}}\right\rangle={\widehat{\widetilde{CZ}}}_{e_4}{\widehat{\widetilde{CZ}}}_{e_3}{\widehat{\widetilde{CZ}}}_{e_2}{\widehat{\widetilde{CZ}}}_{e_1}\left|+_{256}\right\rangle^{\otimes3}$ aligns with the multihypergraph illustrated in Fig.\ref{fig4}(a).
\begin{figure}[h]%
	\centering
	\includegraphics[width=0.45\textwidth]{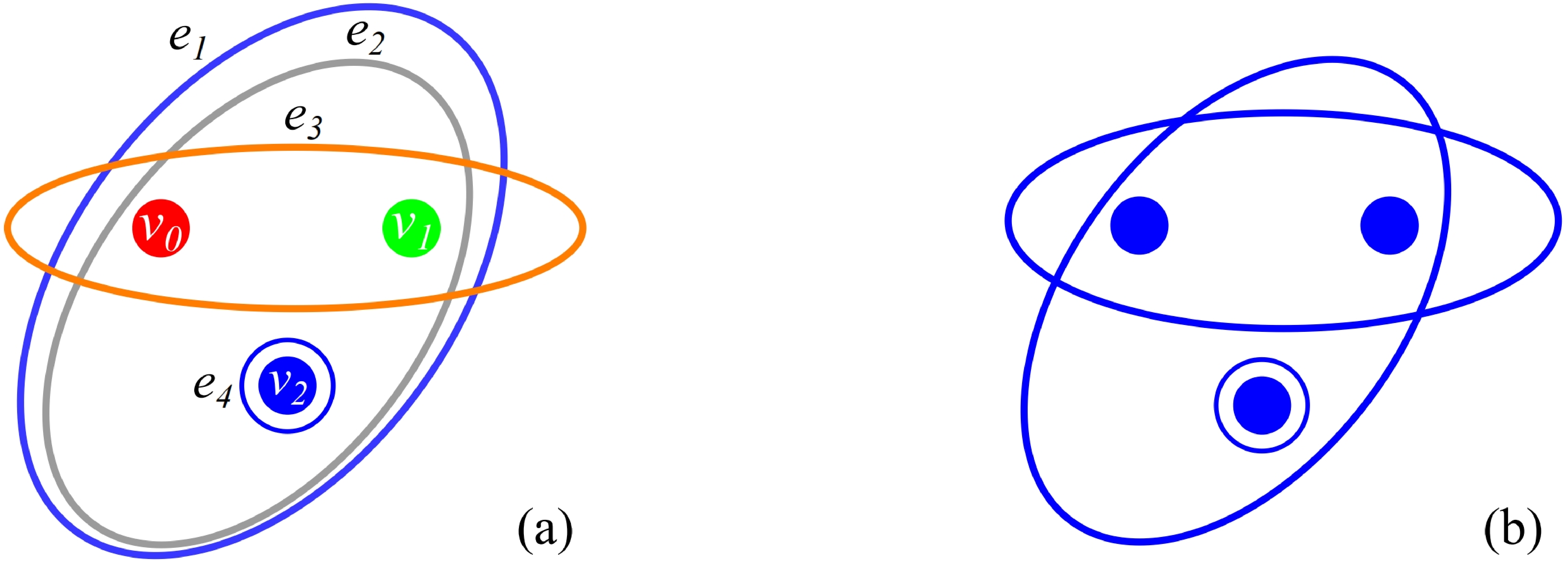}
	\caption{3-vertex multihypergraphs and 3-vertex hypergraphs. (a) shows a multihypergraph with four hyperedges $e_1$ to $e_4$, each with unique RGB values (55,55,255), (155,155,155), (255,125,0) and (0,0,255). (b) shows a hypergraph, where all vertices and the hyperedge uniformly display blue (0,0,255).}\label{fig4}
\end{figure}
\section{\label{sec:4}The relationship between GREWSs and multihypergraph states}
In 2013, Ross et al. \cite{rossi2013quantum} and Qu et al. \cite{qu2013encoding} introduced the concept of hypergraph states, highlighting their one-to-one correspondence with REWSs. Later, in 2018, Xiong et al. \cite{xiong2018qudit} in their research indicated that $d$-dimensional quantum hypergraph states are a subset of GREWSs. Noting that the count of GREWSs is $d^{d^N-1}$, while the number of $d$-dimensional hypergraph states amounts to $d^{2^N-1}$. In this section, firstly, we provide a rigorous proof of the one-to-one correspondence between $2$-dimensional hypergraph states and REWSs. Then, we establish the specific relationship between $d$-dimensional hypergraph states and GREWSs using a similar proof approach, identifying which GREWSs cannot be constructed by $d$-dimensional hypergraph states. Finally, we will elucidate and substantiate the relationship between the proposed multihypergraph states and GREWSs.
\subsection{\label{subsec:3.1}The relationship between 2-dimensional hypergraph states and REWSs}
In this subsection, we firstly state the theorem establishing the bijective relationship between the $2$-dimensional hypergraph state and the REWS. Subsequently, we offer a mathematical proof of this theorem. Consider an arbitrary REWS, denoted as 
\begin{equation}
	\left|f\right\rangle=\frac{1}{2^{N/2}}\sum_{i_0,\cdots,i_{N-1}=0}^{1}{\left(-1\right)^{f\left(i_0,\cdots,i_{N-1}\right)}\left|i_0,\cdots,i_{N-1}\right\rangle},
	\label{eq:12}
\end{equation}
where $f\left(i_0,\cdots{,i}_{N-1}\right)$ is a boolean function of $n$ independent variables $i_0,\cdots{,i}_{N-1} \in \mathbb{Z}_{2}$. By setting the coefficient of the $\left|0,\cdots,0\right\rangle$ term as the global phase and letting $f\left(0,\cdots,0\right)=0$, the total number of REWSs is deduced to be $2^{2^N-1}$. Notably, given that $i_0,\cdots{,i}_{N-1}$ are not simultaneously zero, we define the condition $i_{l_0}i_{l_1},\cdots,i_{l_{t\prime-1}}>0$ to encompass all non-zero terms in $i_0,\cdots{,i}_{N-1}$, where $l_0<l_1<\cdots<l_{t\prime-1}\in\mathbb{Z}_N$, $t\prime\in\mathbb{Z}_{N+1}$.
\begin{theorem}\label{thm1}
	Ror any given hypergraph state $\scriptstyle\left|\widetilde{G}\right\rangle$ that corresponds to a hypergraph $\scriptstyle\widetilde{G}=\left(V,\widetilde E\right)$, it must be a REWS. Furthermore, every REWS $\left|f\right\rangle$ can be associated with a specific hypergraph state $\scriptstyle\left|\widetilde{G}\right\rangle$, such that $\left|f\right\rangle=\scriptstyle\left|\widetilde{G}\right\rangle$.
\end{theorem}
\begin{proof}
Based on the definitions of hypergraph states and REWSs, it is obvious that a hypergraph state invariably constitutes a REWS \cite{rossi2013quantum,qu2013encoding}. Therefore, to substantiate \textbf{Theorem \ref{thm1}}, it suffices to demonstrate that each REWS corresponds to a hypergraph state. Consider a hypergraph state $\scriptstyle\left|\widetilde{G}\right\rangle=\left(\prod_{e\in \widetilde E}{{\widetilde{CZ}}_e}^{m_e}\right)\left|+_d\right\rangle^{\otimes N}$, constructed via $\scriptstyle\left\{m_{e\in \widetilde E}\right\}$ $\scriptstyle\left\{{\widetilde{CZ}}_{e\in \widetilde E}\right\}$ operations, such that $\left|f\right\rangle=\scriptstyle\left|\widetilde{G}\right\rangle$. Analyzing the coefficients $\scriptstyle\left\{\left(-1\right)^{f\left(i_0,\cdots{,i}_{N-1}\right)}\right\}$ of all superposition terms in $\left|f\right\rangle$ allows us to derive the ensuing equations
\begin{multline}
	\fontsize{8}{10}
	\Bigl\{ \sum_{\substack{\left\{ v_0,v_1,\cdots,v_{t-1}\right\} \in\wp(V)\backslash\mathrm{\Phi}, \\ \left\{v_0,v_1,\cdots,v_{t-1}\right\} \subseteq\left\{l_0,l_1,\cdots,l_{t\prime-1}\right\}}}m_{e\left(v_0,v_1,\cdots,v_{t-1}\right)}= \Bigr. \\[-1mm]
	\Bigl. f\left(i_0,\cdots{,i}_{N-1}\right) \bigm| t, t\prime\in\mathbb{Z}_{N+1}^\ast;i_0,\cdots{,i}_{N-1}\in\mathbb{Z}_2; \Bigr. \\[-2mm]
	\Bigl.\left(i_0,\cdots{,i}_{N-1}\right)\neq\left(0,\cdots,0\right) \Bigr\},
	\label{eq:13}
\end{multline}
which is a non-homogeneous linear system with $2^N-1$ independent variables and $2^N-1$ equations in the finite field ${GF}_2$. The equations, denoted as Eq. \eqref{eq:13}, are categorized based on the count of non-zero terms in $i_0,\cdots,i_{N-1}$. The $n$th($n\in\mathbb{Z}_{N+1}^\ast$) group of equations encompasses $C_N^n$ independent variables and $C_N^n$ linear equations. Within the $n$th subset, all non-zero elements among the independent variables $i_0,\cdots,i_{N-1}$ of $f\left(i_0,\cdots,i_{N-1}\right)$ are exclusively $i_{l_0}i_{l_1},\cdots,i_{l_{n-1}}$. The solvability of Eq. \eqref{eq:13} is contingent upon each subgroup within the $N$ sets of equations possessing a solution. Specifically, for $n=1$, the equation set is formulated as
\begin{multline}
	\fontsize{8}{10}
	\Bigl\{ m_{e\left(v_0\right)}=f\left(i_0,\cdots{,i}_{N-1}\right) \Bigm| v_0\in V,i_{v_0}\in\mathbb{Z}_2^\ast;  \\
	i_j=0,j\in V\backslash \left\{v_0\right\} \Bigr\},
	\label{eq:14}
\end{multline}
whose coefficient matrix is the identity matrix $I_{C_N^1\times C_N^1}$, indicating that the equations are solvable. Assuming that the first $n-1$ sets of equations yield solvable, the $n$th set can be simplified to
\begin{multline}
	\fontsize{8}{10}
	\Bigl\{ m_{e\left(v_0,v_1,\cdots,v_{n-1}\right)}=b_{v_0,v_1,\cdots,v_{n-1}} \bigm| n\in\mathbb{Z}_{N+1}^\ast;\\[-1mm] \left\{v_0,v_1,\cdots,v_{n-1}\right\} \in \wp(V)\backslash\mathrm{\Phi}; \Bigr.
	\Bigl.i_{v_0},\cdots,i_{v_{n-1}}\in\mathbb{Z}_2^\ast;\\[-1mm] i_j=0,j\in V\backslash\left\{v_0,v_1,\cdots,v_{n-1}\right\} \Bigr\},
	\label{eq:15}
\end{multline}
whose coefficient matrix is the identity matrix $I_{C_N^n\times C_N^n}$, where $b_{v_0,v_1,\cdots,v_{n-1}}$ represents the value on the right-hand side of each equation post-simplification. Consequently, if solutions exist for the initial $n-1$ sets of equations, then a solution is also assured for the $n$th set. This inference, coupled with the solvability of the first set of equations, enables a recursive determination that all the linear equations are solvable, implying that the Eq.  \eqref{eq:13} are resolvable. Therefore, any REWS $\left|f\right\rangle$ can be equated to a hypergraph state $\scriptstyle\left|\widetilde{G}\right\rangle$, constructed through a specified sequence of $\scriptstyle\left\{m_{e\in \widetilde E}\right\}$ $\scriptstyle\left\{{\widetilde{CZ}}_{e\in \widetilde E}\right\}$ operations, such that $\left|f\right\rangle$$=\scriptstyle\left|\widetilde{G}\right\rangle$. Given that a hypergraph state is inherently a REWS \cite{rossi2013quantum,qu2013encoding},  \textbf{Theorem \ref{thm1}} is substantiated.
\end{proof}
\subsection{The relationship between $d$-dimensional hypergraph states and GREWSs}
In this subsection, the relationship between GREWSs and $d$-dimensional hypergraph states is examined. It is also revealed that the reason why GREWSs are more numerous and the GREWSs which cannot be generated from $d$-dimensional hypergraph states. An arbitrary GREWS is defined as
\begin{equation}
	\left|f_d\right\rangle=\frac{1}{2^{n/2}}\sum_{i_0,\cdots,i_{N-1}=0}^{d-1}{{\omega_d}^{f\left(i_0,\cdots,i_{N-1}\right)}\left|i_0,\cdots,i_{N-1}\right\rangle},
	\label{eq:16}
\end{equation}
where, $f\left(i_0,\cdots,i_{N-1}\right)$ is an integer function of n independent variables $i_0,\cdots,i_{N-1}$ in the integer residual ring $\mathbb{Z}_d$. By isolating the $\left|0,\cdots,0\right\rangle$ term's coefficient as the global phase, setting $f\left(0,\cdots,0\right)=0$, the count of GREWSs is calculated as $d^{d^N-1}$. Consistent with the aforementioned approach, all non-zero terms $i_0,\cdots,i_{N-1}$ are collectively denoted as $i_{l_0},\cdots{,i}_{l_{t\prime-1}}$. Consider a $d$-dimensional hypergraph state $\scriptstyle\left|\widetilde{G}\right\rangle=\left(\prod_{e\in \widetilde E}{{\widetilde{CZ}}_e}^{m_e}\right)\left|+_d\right\rangle^{\otimes N}$, formulated through $\scriptstyle\left\{m_{e\in \widetilde E}\right\}$ $\scriptstyle\left\{{\widetilde{CZ}}_{e\in \widetilde E}\right\}$ operations, such that $\left|f_d\right\rangle$$=\scriptstyle\left|\widetilde{G}\right\rangle$. To analyze the coefficients $\scriptstyle\left\{{\omega_d}^{f\left(i_0,\cdots,i_{N-1}\right)}\right\}$ in all superposition terms of $\left|f_d\right\rangle$, equations
\begin{multline}
	\fontsize{8}{10}
	\Biggl\{ \sum_{\substack{\left\{ v_0,v_1,\cdots,v_{t-1}\right\} \in\wp(V)\backslash\mathrm{\Phi}, \\ \left\{v_0,v_1,\cdots,v_{t-1}\right\} \subseteq\left\{l_0,l_1,\cdots,l_{t\prime-1}\right\}}}\left(\prod_{j}^{t}i_{v_j}\right)m_{e\left(v_0,v_1,\cdots,v_{t-1}\right)}\\[-2mm] = \Biggr.
	\Biggl. f\left(i_0,\cdots{,i}_{N-1}\right) \biggm| t, t\prime\in\mathbb{Z}_{N+1}^\ast;i_0,\cdots{,i}_{N-1}\in\mathbb{Z}_d; \Biggr. \\[-3mm]
	\Biggl.\left(i_0,\cdots{,i}_{N-1}\right)\neq\left(0,\cdots,0\right) \Biggr\}
	\label{eq:17}
\end{multline}
are established, where the operations are conducted modulo $d$. Consequently, the aforementioned equations formulate a linear system with $2^N-1$ independent variables and $d^N-1$ equations in the integer residual ring $\mathbb{Z}_d$. A fundamental requirement for the solvability of a linear system is that the rank of the system's coefficient matrix must equal the rank of its augmented matrix. Thus, for the linear equations Eq. \eqref{eq:17} to be solvable, $\left\{f\left(i_0,\cdots,i_{N-1}\right)\right\}$ must adhere to a specific relational criterion. Given that the augmented matrix of Eq. \eqref{eq:17} has a maximum rank of $2^N-1$, $d$-dimensional hypergraph states can construct at most $d^{2^N-1}$ of the $d^{d^N-1}$ GREWSs. For instance, with $N=2$ and $d=3$, a linear equation system over the finite field ${GF}_3$ can be derived as
\begin{equation}
	\left\{
	\begin{array}{cccc}
		1\times m_{e\left(0\right)} &= f\left(1,0\right) \\
		2\times m_{e\left(0\right)} &= f\left(2,0\right) \\
		1\times m_{e\left(1\right)} &= f\left(0,1\right) \\
		2\times m_{e\left(1\right)} &= f\left(0,2\right) \\
		1\times m_{e\left(0\right)}+1\times m_{e\left(1\right)}+1\times1\times m_{e\left(0,1\right)} &= f\left(1,1\right) \\
		1\times m_{e\left(0\right)}+2\times m_{e\left(1\right)}+1\times2\times m_{e\left(0,1\right)} &= f\left(1,2\right)\\
		2\times m_{e\left(0\right)}+1\times m_{e\left(1\right)}+2\times1\times m_{e\left(0,1\right)} &= f\left(2,1\right) \\
		2\times m_{e\left(0\right)}+2\times m_{e\left(1\right)}+2\times2\times m_{e\left(0,1\right)} &\ = f\left(2,2\right)
	\end{array}
	\right.,
	\label{eq:18}
\end{equation}
which comprises $2^2-1=3$ independent variables and $3^2-1=8$ equations. The system is solvable if and only if $f\left(i_0,i_1\right)$ fulfills
\begin{equation}
	\left\{
	\begin{array}{ccc}
		f\left(1,0\right) &= 2f\left(2,0\right) \\
		f\left(0,1\right) &= 2f\left(0,2\right) \\
		2f\left(1,1\right)+2f\left(1,0\right) &= \ f\left(1,2\right) \\
		2f\left(2,2\right)+f\left(1,0\right) &= \ f\left(2,1\right)
	\end{array}
	\right..
	\label{eq:19}
\end{equation}
The probability amplitude of the following $2$-particle, $3$-dimensional GREWSs $\left|f_3\right\rangle$ fails to conform to the aforementioned equation. Consequently, it cannot be realized via the $d$-dimensional hypergraph state
\begin{multline}
	\fontsize{8}{10}
	\left|f_3\right\rangle=\frac{1}{2}\Bigl[\left|00\right\rangle+\left(e^{{2\pi \mathbf{i}}/{3}}\right)^1\left|01\right\rangle+\left|02\right\rangle+\left(e^{{2\pi \mathbf{i}}/{3}}\right)^1\left|10\right\rangle+\\ \Bigr.
	\Bigl.\left(e^{{2\pi \mathbf{i}}/{3}}\right)^1\left|11\right\rangle+\left|12\right\rangle+\left|20\right\rangle+\left(e^{{2\pi \mathbf{i}}/{3}}\right)^1\left|21\right\rangle+\left|22\right\rangle\Bigr].
	\label{eq:20}
\end{multline}
This study acknowledges that the resolution of Eq. \eqref{eq:17} necessitates consideration of the variable $d$. Specifically, when $d$ equals a prime number $q$, the solution of Eq. \eqref{eq:17} can be determined with relative ease within the finite field ${GF}_q$. Conversely, if $d$ is a composite number, the solution becomes more complex within the integer residual ring $\mathbb{Z}_d$. This intricacy will be elaborated upon in the subsequent subsection. Furthermore, this analysis has led to a formalization of the GREWSs that are unattainable through the construction via $d$-dimensional hypergraph states. This limitation stems from the discrepancy in the linear system associated with the hyperedge count, which consists of $d^N-1$ equations, yet is constrained by only $2^N-1$ independent variables.
\subsection{The relationship between multihypergraph states and GREWSs}
In the preceding subsection, the GREWSs was defined and reasons were delineated for the inability of $d$-dimensional hypergraph states to construct all GREWSs. In this subsection, the linear equations is developed to substantiate the relationship between the proposed multihypergraph states and GREWSs.
\begin{theorem}\label{thm2}
	For any given multihypergraph state $\scriptstyle\left|\widehat{\widetilde{G}}\right\rangle=\left(\prod_{e\in \widehat{\widetilde{E}}}{{\widehat{\widetilde{CZ}}}_e}^{m_e}\right)\left|+\right\rangle^{\otimes N}$ that corresponds to a multihypergraph $\scriptstyle\widehat{\widetilde{G}}=(V,\widehat{\widetilde E})$, It must be a GREWS. Furthermore, For prime $d=q$, every GREWS $\left|f_d\right\rangle$ can be associated with a specific hypergraph state $\scriptstyle\left|\widehat{\widetilde{G}}\right\rangle$, such that $\left|f_d\right\rangle=\scriptstyle\left|\widehat{\widetilde{G}}\right\rangle$.
\end{theorem}
\begin{proof}
	Based on the definitions of multihypergraph states and GREWSs, it is obvious that a multihypergraph state invariably constitutes a GREWS. Therefore, to substantiate \textbf{Theorem \ref{thm2}}, it suffices to demonstrate that each GREWS corresponds to a multihypergraph state. Consider a multihypergraph state $\scriptstyle\left|\widehat{\widetilde{G}}\right\rangle=\left(\prod_{e\in\widehat{\widetilde{E}} }{{\widehat{\widetilde{CZ}}}_e}^{m_e}\right)\left|+_d\right\rangle^{\otimes N}$, constructed through $\scriptstyle\left\{m_{e\in \widehat{\widetilde E}}\right\}$ $\scriptstyle\left\{{\widetilde{CZ}}_{e\in \widehat{\widetilde E}}\right\}$  operations, such that $\left|f_d\right\rangle$$=\scriptstyle\left|\widehat{\widetilde{G}}\right\rangle$. To analyze the coefficients $\left\{{\omega_d}^{f\left(i_0,\cdots,i_{N-1}\right)}\right\}$ of all superposition terms in $\left|f_d\right\rangle$, we formulate the system of equations
	\begin{multline}
		\fontsize{8}{10}
		{\begin{aligned}
		\Biggl\{ \sum_{\substack{\left\{ v_0,v_1,\cdots,v_{t-1}\right\} \in\wp(V)\backslash\mathrm{\Phi}, \\ \left\{v_0,v_1,\cdots,v_{t-1}\right\} \subseteq\left\{l_0,l_1,\cdots,l_{t\prime-1}\right\}}}
		\sum_{s_0,\cdots,s_{t-1}=1}^{d-1} \left(\prod_{j=0}^{t-1}{i_{v_j}}^{s_j}\right) \\[-1mm] \cdot m_{e\left(v_0,v_1,\cdots,v_{t-1},s\right)} = f\left(i_0,\cdots{,i}_{N-1}\right) 
		\biggm| t, t\prime\in\mathbb{Z}_{N+1}^\ast; i_0,\\[-2mm] \cdots{,i}_{N-1}\in\mathbb{Z}_d; \left(i_0,\cdots{,i}_{N-1}\right)\neq\left(0,\cdots,0\right) \Biggr\},
		\end{aligned}}
		\label{eq:21}
	\end{multline}
which is a non-homogeneous linear system with $d^N-1$ independent variables and $d^N-1$ equations in the finite field ${GF}_q$. Consistent with the aforementioned approach, all non-zero terms $i_0,\cdots,i_{N-1}$ are collectively denoted as $i_{l_0},\cdots{,i}_{l_{t\prime-1}}$. The system's equations are segregated into $N$ groups based on the count of non-zero terms in $i_0,\cdots,i_{N-1}$. Each group, denoted as $n$, comprises $C_N^n\cdot\left(d-1\right)^n$ independent variables and $C_N^n\cdot\left(d-1\right)^n$ equations. Specifically, in the $n$th set of equations, the independent variables of $f\left(i_0,\cdots,i_{N-1}\right)$, $i_0,\cdots,i_{N-1}$, with $n$ non-zero values, are represented as $i_{l_0},i_{l_1},\cdots,i_{l_{n-1}}$. Eq. \eqref{eq:21} is solvable if each of the $N$ groups of equations has a solution. For the case $n=1$, the initial set of equations is
	\begin{multline}
		\fontsize{8}{10}
		\Biggl\{ \sum_{s_{v_0}=1}^{d-1}{{i_{v_0}}^{s_{v_0}}\cdot m_{e\left(v_0,s\right)}}=f\left(i_0,\cdots{,i}_{N-1}\right) \Bigm| v_0\in V,\\[-4mm] i_{v_0}\in\mathbb{Z}_d^\ast;  
		i_j=0,j\in V\backslash \left\{v_0\right\} \Bigr\},
		\label{eq:22}
	\end{multline}
	whose coefficient matrix is
	\begin{equation}
		A = 
		I_{C_N^1\times C_N^1}
		\otimes
		\begin{bmatrix}
			1 & 1 & \cdots & 1 \\
			2 & 2^2 & \cdots & 2^{d-1} \\
			\vdots & \vdots & \ddots & \vdots \\
			(d-1) & (d-1)^2 & \cdots & (d-1)^{d-1}
		\end{bmatrix},
		\label{eq:23}
	\end{equation}
	 and the coefficient matrix possesses full rank, $rank(A)=N\times\left(d-1\right)$, ensuring the solvability of these equations. If the first $n-1$ sets of equations can be solvable, the $n$th set can be simplified to 
	\begin{multline}
		\fontsize{8}{10}
		{\begin{aligned}
				\Biggl\{\sum_{s_0,\cdots,s_{n-1}=1}^{d-1} \left(\prod_{j=0}^{n-1}{i_{v_j}}^{s_j}\right) \cdot m_{e\left(v_0,v_1,\cdots,v_{n-1},s\right)} \cdot f\left(i_0, \right. \\[-1mm] i_1,\cdots  \left. {,i}_{N-1}\right)\biggm| n\in\mathbb{Z}_{N+1}^\ast;\left\{v_0,v_1,\cdots,v_{n-1}\right\}\in\wp(V)\backslash \mathrm{\Phi}; \\[-2mm] i_{v_0}, \cdots,i_{v_{n-1}}\in\mathbb{Z}_d^\ast; i_j=0,j\in V\backslash \left\{v_0,v_1,\cdots,v_{n-1}\right\} \Biggr\},
		\end{aligned}}
		\label{eq:24}
	\end{multline}
	whose coefficient matrix can be derived as
	\begin{equation}
		\small
		B = 
		I_{C_N^n\times C_N^n}
		\otimes
		\begin{bmatrix}
			1 & 1 & \cdots & 1 \\
			2 & 2^2 & \cdots & 2^{d-1} \\
			\vdots & \vdots & \ddots & \vdots \\
			(d-1) & (d-1)^2 & \cdots & (d-1)^{d-1}
		\end{bmatrix}^{\otimes n},
		\label{eq:25}
	\end{equation}
	where $b_{v_0,\cdots,v_{n-1},i_{v_0},\cdots i_{v_{n-1}}}$ represents the simplified value to the right of each equation(referred to as $b_{v_0,v_1,\cdots,v_{n-1}}$ in the proof presented in subsection \ref{subsec:3.1}, owing to all non-zero terms in $\mathbb{Z}_d$ being 1).	The coefficient matrix exhibits full rank, $rank(B)=C_N^n\times\left(d-1\right)^n$. Consequently, if solutions exist for the first $n-1$ sets of equations, then the $n$th set also possesses a solution. This, coupled with the solvability of the first set, implies a recursive solution pattern for all linear equations, establishing that Eq. \eqref{eq:21} is solvable. Hence, when $d$ is prime, every GREWS $\left|f_d\right\rangle$ corresponds to a specific multihypergraph state $\scriptstyle\left|\widehat{\widetilde{G}}\right\rangle$, satisfying $\left|f_d\right\rangle=\scriptstyle\left|\widehat{\widetilde{G}}\right\rangle$. Given that a multihypergraph state is inherently a GREWS. This confirms the validity of \textbf{Theorem \ref{thm2}}. 
\end{proof}
In the previous subsection, it was established that the 2-particle, 3-dimensional GREWS $\left|f_3\right\rangle$ cannot be constructed using a $d$-dimensional hypergraph state. However, employing the proposed 2-particle, 3-dimensional multihypergraph state enables the derivation of the following linear equations Eq.  \eqref{eq:26}
\begin{widetext}
		\begin{equation}
			\footnotesize
			\hspace{-6mm}
			\left\{
			\begin{array}{cccccccc}
				1^1\times m_{e\left(0,s=1\right)}+1^2\times m_{e\left(0,s=2\right)}&=f\left(1,0\right)=1 \\
				2^1\times m_{e\left(0,s=1\right)}+2^2\times m_{e\left(0,s=2\right)}&=f\left(2,0\right)=0 \\
				1^1\times m_{e\left(1,s=1\right)}+1^2\times m_{e\left(1,s=2\right)}&=f\left(0,1\right)=1 \\
				2^1\times m_{e\left(1,s=1\right)}+2^2\times m_{e\left(1,s=2\right)}&=f\left(0,2\right)=0 \\
				f\left(1,0\right)+f\left(0,1\right)+1^1\times1^1\times m_{e\left(0,1,s=(1,1)\right)}+1^1\times1^2\times m_{e\left(0,1,s=(1,2)\right)}+1^2\times1^1\times m_{e\left(0,1,s=(2,1)\right)}+1^2\times1^2\times m_{e\left(0,1,s=(2,2)\right)}&=f\left(1,1\right)=1 \\
				f\left(1,0\right)+f\left(0,2\right)+1^1\times2^1\times m_{e\left(0,1,s=(1,1)\right)}+1^1\times2^2\times m_{e\left(0,1,s=(1,2)\right)}+1^2\times2^1\times m_{e\left(0,1,s=(2,1)\right)}+1^2\times2^2\times m_{e\left(0,1,s=(2,2)\right)}&=f\left(1,2\right)=0\\
				f\left(2,0\right)+f\left(0,1\right)+2^1\times1^1\times m_{e\left(0,1,s=(1,1)\right)}+2^1\times1^2\times m_{e\left(0,1,s=(1,2)\right)}+2^2\times1^1\times m_{e\left(0,1,s=(2,1)\right)}+2^2\times1^2\times m_{e\left(0,1,s=(2,2)\right)}&=f\left(2,1\right)=1 \\
				f\left(2,0\right)+f\left(0,2\right)+2^1\times2^1\times m_{e\left(0,1,s=(1,1)\right)}+2^1\times2^2\times m_{e\left(0,1,s=(1,2)\right)}+2^2\times2^1\times m_{e\left(0,1,s=(2,1)\right)}+2^2\times2^2\times m_{e\left(0,1,s=(2,2)\right)}&=f\left(2,2\right)=0 \\
			\end{array}
			\right.,
			\label{eq:26}
		\end{equation}
\end{widetext}
Simplify and resolve
\begin{equation}
	\hspace{-4mm}
	\left\{
	\begin{array}{cccc}
		m_{e\left(0,s=1\right)}&=2 \\
		m_{e\left(0,s=2\right)}&=2 \\
		m_{e\left(1,s=1\right)}&=2 \\
		m_{e\left(1,s=2\right)}&=2 \\
		m_{e\left(0,1,s=(1,1)\right)}&=0 \\
		m_{e\left(0,1,s=(1,2)\right)}&=1\\
		m_{e\left(0,1,s=(2,1)\right)}&=0 \\
		m_{e\left(0,1,s=(2,2)\right)}&=1
	\end{array}
	\right..
	\label{eq:27}
\end{equation}
Consequently, by executing two $\scriptstyle\sum_{i_0,i_1=0}^{d-1}{\omega_d^{i_0}\left|i_0,i_1\right\rangle\left\langle i_0,i_1\right|}$ operations, two $\scriptstyle\sum_{i_0,i_1=0}^{d-1}{\omega_d^{{i_0}^2}\left|i_0,i_1\right\rangle\left\langle i_0,i_1\right|}$ operations, two $\scriptstyle\sum_{i_0,i_1=0}^{d-1}$${\omega_d^{i_1}\left|i_0,i_1\right\rangle\left\langle i_0,i_1\right|}$ operations, two $\scriptstyle\sum_{i_0,i_1=0}^{d-1}$$\scriptstyle\omega_d^{{i_1}^2}\left|i_0,i_1\right\rangle$\\$\scriptstyle\left\langle i_0,i_1\right|$ operations, one $\scriptstyle\sum_{i_0,i_1=0}^{d-1}{\omega_d^{i_0\times{i_1}^2}\left|i_0,i_1\right\rangle\left\langle i_0,i_1\right|}$ operation and one  $\scriptstyle\sum_{i_0,i_1=0}^{d-1}{\omega_d^{{i_0}^2\times{i_1}^2}\left|i_0,i_1\right\rangle\left\langle i_0,i_1\right|}$ operation on the $\scriptstyle{\left|+_3\right\rangle^{\otimes2}=3}^{-1}\sum_{i,j=0}^{2}\left|ij\right\rangle$ state, a two-particle, three-dimensional GREWS $\left|f_3\right\rangle$ can be constructed. This state, as indicated in the prior subsection, cannot be derived from a $d$-dimensional hypergraph state.

Next, we persist in examining the nexus between the GREWSs and the proposed multihypergraph states, in scenarios where $d$ is a composite. By selecting $d=4$ and $d=6$ as exemplary cases, we aim to derive the coefficient matrix pertinent to
\begin{equation}
	\small
	I_{C_N^n\times C_N^n}
	\otimes
	\begin{bmatrix}
		1 & 1  & 1 \\
		2 & 0 & 0 \\
		3 & 1 & 3
	\end{bmatrix}^{\otimes n}
	\text{and } 
	I_{C_N^n\times C_N^n}
	\otimes
	\begin{bmatrix}
		1 & 1 & 1 & 1 & 1 \\
		2 & 4 & 4 & 4 & 2 \\
		3 & 3 & 3 & 3 & 3 \\
		4 & 4 & 4 & 4 & 4 \\
		5 & 1 & 5 & 1 & 5
	\end{bmatrix}^{\otimes n}.
	\label{eq:28}
\end{equation}
In the case of $d=4$, the element $2$ lacks a multiplicative inverse within the integer residual ring $\mathbb{Z}_4$. Similarly, when $d=6$, the elements $2$, $3$ and $4$ are devoid of multiplicative inverses in the integer residual ring $\mathbb{Z}_6$. Under these conditions, the equation system denoted as Eq. \eqref{eq:24} inevitably encompasses a subset that is unsolvable, implying that certain aspects of the GREWSs cannot be realized through multihypergraph states. For instance, with $d=4$ and $N=1$, it is infeasible to construct the specified GREWSs by multihypergraph states, as the corresponding system of linear equations
\begin{multline}
	\left|f_4\right\rangle=\frac{1}{2}(\left|0\right\rangle+e^{{2\pi \mathbf{i}}/{4}}\left|1\right\rangle+e^{{2\pi \mathbf{i}}/{4}}\left|2\right\rangle+\left(e^{{2\pi \mathbf{i}}/{4}}\right)^2\left|3\right\rangle),\\[-2mm]
	\left\{
	\begin{array}{cccc}
		1\cdot m_{e\left(0,s=1\right)}+1\cdot m_{e\left(0,s=2\right)}+1\cdot m_{e\left(0,s=2\right)}&=1 \\
		2\cdot m_{e\left(0,s=1\right)}+0\cdot m_{e\left(0,s=2\right)}+0\cdot m_{e\left(0,s=2\right)}&=1 \\
		3\cdot m_{e\left(0,s=1\right)}+1\cdot m_{e\left(0,s=2\right)}+3\cdot m_{e\left(0,s=2\right)}&\ =2
	\end{array}
	\right.
	\label{eq:29}
\end{multline}
is insoluble. Based on the preceding analysis, it is discerned that the count of multihypergraph states is equivalent to the number of GREWSs, with both being $d^{d^{N-1}}$. Furthermore, it is intrinsic that multihypergraph states adhere to the GREWSs framework. Consequently, this necessitates that a portion of these multihypergraph states invariably culminates in the formation of identical GREWSs. Illustratively, by considering the scenario where $d=4$ and $N=1$, we observe
\begin{multline}
	\left|f_4\right\rangle=\frac{1}{2}(\left|0\right\rangle+e^{{2\pi \mathbf{i}}/{4}}\left|1\right\rangle+\left(e^{{2\pi \mathbf{i}}/{4}}\right)^2\left|2\right\rangle+e^{{2\pi \mathbf{i}}/{4}}\left|3\right\rangle),\\[-2mm]
	\left\{
	\begin{array}{cccc}
		1\cdot m_{e\left(0,s=1\right)}+1\cdot m_{e\left(0,s=2\right)}+1\cdot m_{e\left(0,s=2\right)}&=1 \\
		2\cdot m_{e\left(0,s=1\right)}+0\cdot m_{e\left(0,s=2\right)}+0\cdot m_{e\left(0,s=2\right)}&=2 \\
		3\cdot m_{e\left(0,s=1\right)}+1\cdot m_{e\left(0,s=2\right)}+3\cdot m_{e\left(0,s=2\right)}&=1 
	\end{array}
	\right..
	\label{eq:30}
\end{multline}
The tuple $\left(m_{e\left(1,s=1\right)},m_{e\left(1,s=2\right)},m_{e\left(1,s=3\right)}\right)$ admits two distinct solutions, $(1,3,1)$ and $(3,1,1)$, indicating that two separate multihypergraph states correspond to an identical GREWS. This observation leads to the inference that when $d$ is a composite number, the state of the multigraph is effectively a subset of the GREWS, with certain multigraph states converging upon the same GREWS. The relationship between the proposed multihypergraph states and the GREWSs is visually represented in Fig.\ref{fig5}.
\begin{figure}[!h]%
	\includegraphics[width=0.45\textwidth]{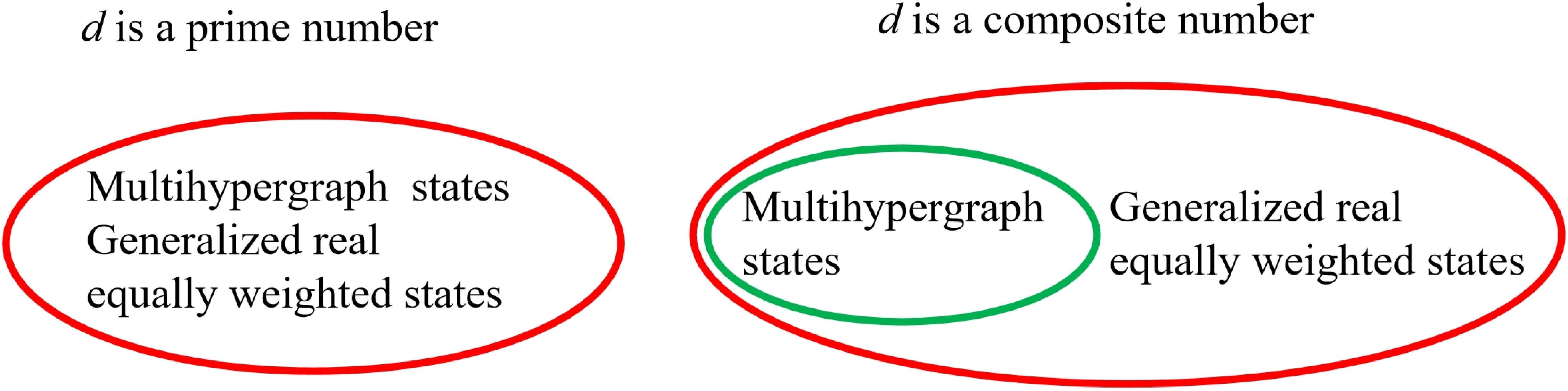}
	\caption{The relationship between the proposed multihypergraph states and the GREWSs.}\label{fig5}
\end{figure}
\section{\label{sec:5}Conclusion}
In this paper, we propose two classes of multiparticle entangled states, the multigraph states and the multihypergraph states, each corresponding to their respective constructs in graph theory. We employ the method of linear equation construction and solution to establish the one-to-one correspondence between the two-dimensional hypergraph states and the REWSs. With the same method, the one-to-one correspondence between the proposed multihypergraph states and the GREWSs, for prime $d$, is also established. Additionally, we highlight the limitations of GREWSs that cannot be generated from $d$-dimensional hypergraph states. The quantum states, we proposed, are constructed by a series of quantum gates, some derived from existing gate constructions. For broader cases, further investigation is anticipated. Given the extensive use of existing graph states and hypergraph states in quantum information and computing, our work suggests that the properties of multigraph states and multihypergraph states may surpass those of traditional graph and hypergraph states. Consequently, the potential application of the proposed multigraph states and multihypergraph states in corresponding fields merits further exploration.

\begin{acknowledgments}
This work was supported by National Natural Science Foundation of China (Grants No. 62171131, 61976053, and 61772134), Fujian Province Natural Science Foundation (Grant No. 2022J01186 and 2023J01533), and Program for New Century Excellent Talents in Fujian Province University.
\end{acknowledgments}
\appendix
\section{\label{ap:A}DERIVATION OF EQ.\eqref{eq:5}}
\begin{proof}
	
If $k\notin e$, 
$\scriptstyle {{\widetilde{CZ}}_e}^{m_e}X_k{{\widetilde{CZ}}_e}^{d-m_e}=X_k$.

If $k\in e$,

$\scriptstyle {{\widetilde{CZ}}_e}^{m_e}X_k{{\widetilde{CZ}}_e}^{d-m_e}=X_k{X_k}^{d-1}{{\widetilde{CZ}}_e}^{m_e}X_k{{\widetilde{CZ}}_e}^{d-m_e}$.

Since
\begin{multline}
	\footnotesize 
	\begin{aligned}
	{\widetilde{CZ}}_{e\left(v_0,v_1,\cdots,v_{t-1}\right)}=\sum_{i_0,i_1,{\cdots,i}_{N-1}=0}^{d-1}\Bigl(\omega_d^{i_{v_0}\times i_{v_1}\times\cdots\times i_{v_{t-1}}}\Bigr. \\[-1.5mm] \Bigl.\cdot\left|i_0,i_1,{\cdots,i}_{N-1}\right\rangle\left\langle i_0,i_1,{\cdots,i}_{N-1}\right|\Bigr),
	\end{aligned}
	\label{eq:A1}
\end{multline}
\begin{widetext}
	\begin{multline}
		\footnotesize 
		\begin{aligned}
		&{X_k}^{d-1}{{{\widetilde{CZ}}_e}^{m_e}X}_k\\[-2mm] &={X_k}^{d-1}\sum_{i_0,i_1,{\cdots,i}_{N-1}=0}^{d-1}\left(\omega_d\right)^{i_{v_0}\times\cdots\times i_{v_{t-1}}\times m_e}\cdot\left|i_0,i_1,\cdots, i_{N-1}\right\rangle\left\langle i_0,i_1,{\cdots,i}_{N-1}\right|X_k\\[-2mm] &=\sum_{i_0,i_1,{\cdots,i}_{N-1}=0}^{d-1}\left(\omega_d\right)^{i_{v_0}\times i_{v_1}\times\cdots\times i_{v_{t-1}}\times m_e}\left|i_0,{\cdots,i_{k-1},i_k+d-1,i_{k+1},\cdots,i}_{N-1}\right\rangle\left\langle i_0,{\cdots,i_{k-1},i_k+1,i_{k+1},\cdots,i}_{N-1}\right|,
		\end{aligned}
		\label{eq:A2}
	\end{multline}
	\begin{multline}
		\footnotesize 
		\begin{aligned}
			&{X_k}^{d-1}{{\widetilde{CZ}}_e}^{m_e}X_k{{\widetilde{CZ}}_e}^{{d-m}_e}\\
			&=\sum_{i_0,i_1,{\cdots,i}_{N-1}=0}^{d-1}\left(\omega_d\right)^{i_{v_0}\times i_{v_1}\times\cdots\times i_{v_{t-1}}\times m_e}\left| i_0,i_1,{\cdots,i_{k-1},i_k+d-1,i_{k+1},\cdots,i}_{N-1}\right\rangle\left\langle i_0,i_1,{\cdots,i_{k-1},i_k-1,i_{k+1},\cdots,i}_{N-1}\right| \\[-3mm]  &\ \ \ \ \ \ \ \ \ \ \ \ \ \ \ \cdot\sum_{{i\prime}_0,{i\prime}_1,{\cdots,i\prime}_{N-1}=0}^{d-1}{\left(\omega_d\right)^{{i\prime}_{v_0}\times{i\prime}_{v_1}\times\cdots\times{i\prime}_{v_{t-1}}({d-m}_e)}\left|{i\prime}_0,{i\prime}_1,{\cdots,i\prime}_{N-1}\right\rangle\left\langle{i\prime}_0,{i\prime}_1,{\cdots,i\prime}_{N-1}\right|}\\[-2mm] &=\sum_{i_0,i_1,{\cdots,i}_{N-1}=0}^{d-1}{\left(\omega_d\right)^{i_{v_0}\times i_{v_1}\times\cdots\times(i_{v_r}+d-1)\times\cdots\times i_{v_{t-1}}\times m_e+i_{v_0}\times i_{v_1}\times\cdots\times i_{v_{t-1}}({d-m}_e)}\left|i_0,i_1,{\cdots,i}_{N-1}\right\rangle\left\langle i_0,i_1,{\cdots,i}_{N-1}\right|}\\[-2mm]
			&=\sum_{i_0,i_1,{\cdots,i}_{N-1}=0}^{d-1}{\left(\omega_d\right)^{i_{v_0}\times i_{v_1}\times\cdots\times(i_{v_r}+d-1-i_{v_r})\times\cdots\times i_{v_{t-1}}\times m_e+i_{v_0}\times i_{v_1}\times\cdots\times i_{v_{t-1}}\times d}\left|i_0,i_1,{\cdots,i}_{N-1}\right\rangle\left\langle i_0,i_1,{\cdots,i}_{N-1}\right|}\\[-2mm]
			&=\sum_{i_0,i_1,{\cdots,i}_{N-1}=0}^{d-1}{\left(\omega_d\right)^{i_{v_0}\times i_{v_1}\times\cdots\times i_{v_{r-1}}\times(d-1){\times i}_{v_{r+1}}\times\cdots\times i_{v_{t-1}}\times m_e}\left|i_0,i_1,{\cdots,i}_{N-1}\right\rangle\left\langle i_0,i_1,{\cdots,i}_{N-1}\right|}={{\widetilde{CZ}}_{e\backslash \left\{k\right\}}}^{m_e(d-1)},v_r=k.
		\end{aligned}
		\label{eq:A3}
	\end{multline}
\end{widetext}
	Then we have $\scriptstyle X_k{X_k}^{d-1}{{\widetilde{CZ}}_e}^{m_e}X_k{{\widetilde{CZ}}_e}^{d-m_e}=X_k{{\widetilde{CZ}}_{e\backslash\left\{k\right\}}}^{m_e(d-1)}$ and\\ $\scriptstyle g_k=\left(\prod_{e\in \widetilde E}{{\widetilde{CZ}}_e}^{m_e}\right)X_k\left(\prod_{e\prime\in \widetilde E}{{\widetilde{CZ}}_{e\prime}}^{{d-m}_{e\prime}}\right)=X_k\prod_{k\in e}\scriptstyle{{\widetilde{CZ}}_{e\backslash \left\{k\right\}}}^{m_e(d-1)}$.\\$\ $
\end{proof}
\section{\label{ap:B}DERIVATION OF EQ.\eqref{eq:11}}
\begin{proof}
If $k\notin e$, 
$\scriptstyle {{\widehat{\widetilde{CZ}}}_e}^{m_e}X_k{{\widehat{\widetilde{CZ}}}_e}^{{d-m}_e}=X_k$.

If $k\in e$,

$\scriptstyle {{\widehat{\widetilde{CZ}}}_e}^{m_e}X_k{{\widehat{\widetilde{CZ}}}_e}^{{d-m}_e}=X_k{X_k}^{d-1}{{\widehat{\widetilde{CZ}}}_e}^{m_e}X_k{{\widehat{\widetilde{CZ}}}_e}^{{d-m}_e}$.

Since
\begin{multline}
	\footnotesize 
	\begin{aligned}
		{\widehat{\widetilde{CZ}}}_e=\sum_{i_0,i_1,{\cdots,i}_{N-1}=0}^{d-1}\Bigl( \omega_d^{{i_{v_0}}^{s_{v_0}}\times{i_{v_1}}^{s_{v_1}}\times{\cdots\times i_{v_{t-1}}}^{s_{v_{t-1}}}} \\[-2mm]  \cdot\left|i_0,i_1,{\cdots,i}_{N-1}\right\rangle\left\langle i_0,i_1,\cdots,i_{N-1}\right|\Bigr),
	\end{aligned}
	\label{eq:B1}
\end{multline}
\begin{widetext}
	\begin{multline}
		 \tiny
		\begin{aligned}
		&{X_k}^{d-1}{{\widehat{\widetilde{CZ}}}_e}^{m_e}X_k=\sum_{i_0,i_1,{\cdots,i}_{N-1}=0}^{d-1}\Bigl( \omega_d^{m_e\times{i_{v_0}}^{s_{v_0}}\times{i_{v_1}}^{s_{v_1}}\times{\cdots\times i_{v_{t-1}}}^{s_{v_{t-1}}}}\Bigr. \\[-4mm] \Bigl. &\ \ \ \ \ \ \ \ \ \ \ \ \ \ \ \ \ \ \ \ \ \ \ \ \ \ \ \ \ \ \ \ \ \ \ \ \ \ \ \ \ \ \ \ \ \ \ \ \ \ \ \ \ \ \ \ \ \ \ \ \ \ \ \ \cdot\left|i_0,i_1,{\cdots,i_{k-1},i_k+d-1,i_{k+1},\cdots,i}_{N-1}\right\rangle\left\langle i_0,i_1,{\cdots,i_{k-1},i_k-1,i_{k+1},\cdots,i}_{N-1}\right|\Bigr),
		\end{aligned}
		\label{eq:B2}
	\end{multline}
	\begin{multline}
		\tiny 
		\begin{aligned}
			&{X_k}^{d-1}{{\widehat{\widetilde{CZ}}}_e}^{m_e}X_k{{\widehat{\widetilde{CZ}}}_e}^{{d-m}_e}
			\\ &=\sum_{i_0,i_1,{\cdots,i}_{N-1}=0}^{d-1}\omega_d^{m_e\times{i_{v_0}}^{s_{v_0}}\times{i_{v_1}}^{s_{v_1}}\times{\cdots\times i_{v_{t-1}}}^{s_{v_{t-1}}}} \left|i_0,i_1,{\cdots,i_{k-1},i_k+d-1,i_{k+1},\cdots,i}_{N-1}\right\rangle \left\langle i_0,i_1,\cdots,i_{k-1},i_k-1,\right.\\[-4mm] &\left. \ \ \ \ \ \ \ \ \ i_{k+1}, \cdots,i_{N-1}\right|\cdot\sum_{{i\prime}_0,{i\prime}_1,{\cdots,i\prime}_{N-1}=0}^{d-1}{\left(\omega_d\right)^{(d-m_e)\times{{i\prime}_{v_0}}^{s_{v_0}}\times{{i\prime}_{v_1}}^{s_{v_1}}\times{\cdots\times{i\prime}_{v_{t-1}}}^{s_{v_{t-1}}}}\left|{i\prime}_0,{i\prime}_1,{\cdots,i\prime}_{N-1}\right\rangle\left\langle{i\prime}_0,{i\prime}_1,{\cdots,i\prime}_{N-1}\right|}\\[-2mm]
			&=\sum_{i_0,i_1,{\cdots,i}_{N-1}=0}^{d-1}{\left(\omega_d\right)^{m_e\times{i_{v_0}}^{s_{v_0}}\times\cdots\times({i_{v_r}}^{s_{v_r}}+d-1)\times\cdots\times{i_{v_{t-1}}}^{s_{v_{t-1}}}+(d-m_e)\times{i_{v_0}}^{s_{v_0}}\times\cdots\times{i_{v_{t-1}}}^{s_{v_{t-1}}}}\left|i_0,{\cdots,i}_{N-1}\right\rangle\left\langle i_0,{\cdots,i}_{N-1}\right|}\\[-2mm]
			&=\sum_{i_0,i_1,{\cdots,i}_{N-1}=0}^{d-1}{\left(\omega_d\right)^{m_e\times{i_{v_0}}^{s_{v_0}}\times\cdots\times({i_{v_r}}^{s_{v_r}}+d-1-{i_{v_r}}^{s_{v_r}})\times\cdots\times{i_{v_{t-1}}}^{s_{v_{t-1}}}+d\times{i_{v_0}}^{s_{v_0}}\times\cdots\times{i_{v_{t-1}}}^{s_{v_{t-1}}}}\left|i_0,{\cdots,i}_{N-1}\right\rangle\left\langle i_0,{\cdots,i}_{N-1}\right|}\\[-2mm]
			&=\sum_{i_0,i_1,{\cdots,i}_{N-1}=0}^{d-1}{\left(\omega_d\right)^{m_{e_c}\times{i_{v_0}}^{s_{v_0}}\times\cdots\times{i_{v_{r-1}}}^{s_{v_{r-1}}}\times(d-1)\times{i_{v_{r+1}}}^{s_{v_{r+1}}}\times\cdots\times{i_{v_{t-1}}}^{s_{v_{t-1}}}}\left|i_0,i_1,{\cdots,i}_{N-1}\right\rangle\left\langle i_0,i_1,{\cdots,i}_{N-1}\right|},v_r=k.
		\end{aligned}
			\label{eq:B3}
	\end{multline}
\end{widetext}
	Then we get $\scriptstyle X_k{X_k}^{d-1}{{\widehat{\widetilde{CZ}}}_e}^{m_e}X_k{{\widehat{\widetilde{CZ}}}_e}^{{d-m}_e}=X_k{{\widehat{\widetilde{CZ}}}_{e\backslash\left\{k\right\}}}^{m_e(d-1)}$and
	$\scriptstyle g_k=$\\$\scriptstyle\left(\prod_{e\in \widehat{\widetilde E}}{{\widehat{\widetilde{CZ}}}_e}^{m_e}\right)X_k\left(\prod_{e\prime\in \widehat{\widetilde E}}{{\widehat{\widetilde{CZ}}}_{e\prime}}^{{d-m}_{e\prime}}\right)=X_k\prod_{e\in \widehat{\widetilde E},k\in e}{{\widehat{\widetilde{CZ}}}_{e\backslash\left\{k\right\}}}^{m_e(d-1)}$.\\$\  $
\end{proof}

\bibliography{references}
\end{document}